\documentclass[runningheads]{llncs}

\usepackage[colorinlistoftodos]{todonotes}
\usepackage[colorlinks=true, allcolors=blue]{hyperref}
\usepackage{amsmath}
\usepackage{amsfonts}
\usepackage{subcaption}
\usepackage{graphicx}
\usepackage[ruled,vlined,french]{algorithm2e}
\usepackage{tikz}
\usepackage{float} 
\usepackage[justification=centering]{caption}
\usepackage{caption}
\usepackage{enumitem}
\usepackage{todonotes}

\newtheorem{defn}{Definition}

\newtheorem{lemme}[defn]{Lemma}
\newtheorem{cor}[defn]{Corollary}

\newcommand{\T}[1]{T(#1)}
\newcommand{\I}[1]{ I_{#1}=(X_{i_{#1}}, S_{I_{#1}}, \int {I_{#1}},\ext {I_{#1}},\paire {I_{#1})}}

\renewcommand{\int}[1]{D_{int}{{(#1)}}}
\newcommand{\ext}[1]{D_{ext}{{(#1)}}}
\newcommand{\paire}[1]{D_{pair}{{(#1)}}}

\newcommand{\dimx}[1]{\dim(#1)}

\newcommand{\ddeux}[1]{d_{\leq 2}(#1)}
\newcommand{\distvec}[2]{\vec{\mathbf{{d_{#1}}(#2)}}}

\renewcommand{\vec}{\mathbf}

\begin{document}
\title{Metric dimension parameterized by treewidth in chordal graphs\thanks{This work was supported by ANR project GrR (ANR-18-CE40-0032)}}
\author{Nicolas Bousquet\inst{1} \and Quentin Deschamps\inst{1} \and Aline Parreau\inst{1}}
\institute{Univ. Lyon, Universit\'e Lyon 1, CNRS, LIRIS UMR 5205, F-69621, Lyon, France.}
\maketitle

\begin{abstract}
The metric dimension has been introduced independently by Harary, Melter~\cite{harray1975} and Slater~\cite{slater1975} in 1975 to identify vertices of a graph $G$ using its distances to a subset of vertices of $G$. A \emph{resolving set} $X$ of a graph $G$ is a subset of vertices such that, for every pair $(u,v)$ of vertices of $G$, there is a vertex $x$ in $X$ such that the distance between $x$ and $u$ and the distance between $x$ and $v$ are distinct. The metric dimension of the graph is the minimum size of a resolving set. Computing the metric dimension of a graph is NP-hard even on split graphs and interval graphs. 

Bonnet and Purohit~\cite{bonnet2021} proved that the metric dimension problem is W[1]-hard parameterized by treewidth. Li and Pilipczuk strenghtened this result by showing that it is NP-hard for graphs of treewidth $24$ in~\cite{li2022}.
In this article, we prove that that metric dimension is FPT parameterized by treewidth in chordal graphs.
\end{abstract}

\section{Introduction}\label{intro}

Determining the position of an agent on a network is a central problem. One way to determine his position is to place sensors on nodes of the network and the agents try to determine their positions using their positions with respect to these sensors. More formally, assume that the agent knows the topology of the graph. Can he, by simply looking at his position with respect to the sensors determine for sure his position in the network? Conversely, where do sensors have to be placed to ensure that any agent at any possible position can easily determine for sure its position? These questions received a considerable attention in the last decades and have been studied in combinatorics under different names such as metric dimension, identifying codes, locating dominating sets...

Let $G=(V,E)$ be a graph and $s,u,v$ be three vertices of $G$. We say that $s$ \emph{resolves} the pair $(u,v)$ if the distance between $s$ and $u$ is different from the distance between $s$ and $v$.
A \emph{resolving set} of a graph $G=(V,E)$ is a subset $S$ of vertices of $G$ such that any vertex of $G$ is identified by its distances to the vertices of the resolving set. In other words, $S$ is a resolving set if for every pair $(u,v)$ of vertices of $G$, there is a vertex $s$ of $S$ such that $s$ resolves $(u,v)$.  
The \emph{metric dimension} of $G$, denoted by $\dim(G)$, is the smallest size of a resolving set of $G$.
This notion has been introduced in 1975 by Slater \cite{slater1975} for trees and by Harary and Melter \cite{harray1975} for graphs to simulate the moves of a sonar. The \emph{metric dimension} of $G$ is the smallest size of a resolving set of $G$. The associated decision problem, called the \textsc{Metric Dimension} problem, is defined as follows: given a graph $G$ and an integer $k$, is the metric dimension of~$G$ is at most $k$?

The \textsc{Metric Dimension} problem is NP-complete~\cite{garey1979} even for restricted classes of graphs like planar graphs~\cite{Diaz2012}.
Epstein et al.~\cite{epstein2015} proved that this problem is NP-complete on split graphs, bipartite and co-bipartite graphs. The problem also is NP-complete on interval graphs~\cite{foucaud2017} or sub-cubic graphs~\cite{hartung2013}. On the positive side, computing the metric dimension is linear on trees~\cite{harray1975,slater1975} and polynomial in outer-planar graphs~\cite{Diaz2012}. 

\paragraph{Parameterized algorithms.} 
In this paper, we consider the \textsc{Metric Dimension} problem from a parameterized point of view.  We say a problem $\Pi$ is \emph{fixed parameter tractable} (FPT) for a parameter $k$ if any instance of size $n$ and parameter $k$ can be decided in time $f(k) \cdot n^{O(1)}$. Two types of parameters received a considerable attention in the litterature: the size of the solution and the "width" of the graph (for various widths, the most classical being the treewidth). 

Hartung and Nichterlein proved in~\cite{hartung2013} that the \textsc{Metric Dimension} problem is W[2]-hard parameterized by the size of the solution. Foucaud et al. proved that it is FPT parameterized by the solution size in interval graphs in~\cite{foucaud2017}. This result was extended by Belmonte et al. who proved in~\cite{BelmonteFGR16} that \textsc{Metric Dimension} is FPT parameterized by the size of the solution plus the tree-length of the graph. In particular, it implies that computing the metric dimension for chordal graph is FPT parameterized by the size of the solution.

\textsc{Metric Dimension} is FPT paramerized by the modular width~\cite{BelmonteFGR16}. Using Courcelle's theorem, one can also remark that it is FPT paramerized by the treedepth of the graph as observed in~\cite{gima2022}.
\textsc{Metric dimension} has been proven W[1]-hard parameterized by the treewidth by Bonnet and Purohit in~\cite{bonnet2021}. Li and Pilipczuk strenghtened this result by showing that it is NP-complete for graphs of treewidth, and even pathwidth, $24$ in~\cite{li2022}. While \textsc{Metric dimension} is polynomial on graphs of treewidth~$1$ (forests), its complexity is unknown for graphs of treewidth~$2$ is open (even if it is known to be polynomial for outerplanar graphs). 
Our main result is the following:

\begin{theorem}\label{thm:main}
\textsc{Metric Dimension} is FPT parameterized by treewidth on chordal graphs. That is, \textsc{Metric Dimension} can be decided in time $O(n^3+n^2 \cdot f(\omega))$ on chordal graphs of clique number $\omega$.
\end{theorem}

Recall that, on chordal graphs, the treewidth is equal to the size of a maximum clique minus one.
Our proof is based on a dynamic programming algorithm.
One of the main difficulty to compute the metric dimension is that a pair of vertices might be resolved by a vertex far from them in the graph. This non-locality, implies that it is not simple to use classical algorithmic strategies like divide-and-conquer, induction or dynamic programming since a single edge or vertex modification somewhere in the graph might change the whole solution\footnote{The addition of a single edge in a graph might modify the metric dimension by $\Omega(n)$, see e.g.~\cite{eroh2015}.}. 

The first ingredient of our algorithm consists in proving that, given a chordal graph, if we are using a clique tree of a desirable form and make some simple assumptions on the shape of an optimal solution, we can ensure that resolving a pairs of vertices close to a separator implies that we resolve all the pairs of vertices in the graph. 
Using this lemma, we build a dynamic programming algorithm 
 that computes the minimum size of a resolving set containing a given vertex in FPT-time parameterized by treewdith.

The special type of clique tree used in the paper, inspired from~\cite{kloks1994}, is presented in Section~\ref{sec:treedec}. We then give some properties of resolving sets  in chordal graphs in Section~\ref{sec:chordal}. These properties will be needed to prove the correction and the running time of the algorithm. Then, we present the definition of the extended problem in Section~\ref{sec:defpb} and the rules of the dynamic programming in Section~\ref{sec:rules} where we also prove the correction of the algorithm. We end by an analysis of the complexity of the algorithm in Section~\ref{sec:complexity}.

\paragraph{Further work.}
The function of the treewidth in our algorithm is probbly not optimal and we did not try to optimize it to keep the algorithm as simple as possible. A first natural question is the existence of an algorithm running in time $2^{\omega} \cdot Poly(n)$ for chordal graphs.

We know that Theorem~\ref{thm:main} cannot be extended to bounded treewidth graphs since \textsc{Metric Dimension} is NP-hard on graphs of treewidth at most $24$~\cite{li2022}. One can nevertheless wonder if our proof technique can be adapted to design polynomial time algorithms for graphs of treewidth at most $2$ on which the complexity status of \textsc{Metric Dimension} is still open. 

Our proof nevertheless crucially relies on the fact that a separators $X$ of a chordal graphs is a clique and then the way a vertex in a component of $G\setminus X$ is interacting with vertices in another component of $G \setminus X$ is simple. One can wonder if there is a tree decomposition in $G$ where all the bags have diameter at most $C$, is it true that \textsc{Metric Dimension} is FPT parameterized by the size of the bags plus $C$.  Note that, since \textsc{Metric Dimension} is NP-complete on chordal graphs, the problem is indeed hard parameterized by the diameter of the bags only.

\section{Preliminaries}

\subsection{Clique trees}\label{sec:treedec}

Unless otherwise stated, all the graphs considered in this paper are undirected, simple, finite and  connected. For standard terminology and notations on graphs, we refer the reader to~\cite{bookgraph}. Let us first define some notations we use throughout the article.

Let $G=(V,E)$ be a graph where $V$ is the set of vertices of $G$ and $E$ the set of edges; we let $n=|V|$. For two vertices $x$ and $y$ in $G$, we denote by $d(x,y)$ the length of a shortest path between $x$ and $y$ and call it \emph{distance between $x$ and $y$}. For every $x \in V$ and $U \subseteq V$, the  \emph{distance between $x$ and $U$}, denoted by $d(x,U)$, is the minimum distance between $x$ and a vertex of $U$.
Two vertices $x$ and $y$ are \emph{adjacent} if $xy \in E$. A \emph{clique} is a graph where all the pairs of vertices are adjacent. We denote by $\omega$ the size of a maximum clique.
Let $U$ be a set of vertices of $G$. We denote by $G \setminus U$ the subgraph of $G$ induced by the set of vertices $V \setminus U$. We say that $U$ is a \emph{separator} of $G$ if $G \setminus U$ is not connected. If two vertices $x$ and $y$ of $V \setminus U$ belong to two different connected components in $G'$, we say that $U$ \emph{separates} $x$ and $y$. If the set $U$ induces a clique, we say that $U$ is a \emph{clique separator} of $G$.



\begin{defn}\label{def_tree_decom}
A \emph{tree-decomposition} of a graph $G$ is a pair $(X , T)$ where $T$ is a tree and $X = \{X_i| i \in V (T)\}$ is a collection of subsets (called bags) of $V (G)$ such that:
\begin{itemize}
\item $ \bigcup_{i  \in V (T)} X_i = V (G)$.
\item For each edge $xy \in E(G), x, y \in X_i$ for some $i \in V (T)$.
\item For each $x \in V (G)$, the set $\{i | x \in X_i\}$ induces a connected sub-tree of $T$.
\end{itemize}
\end{defn}

Let $G$ be a graph and  $(X,T)$ a tree decomposition of $G$. The \emph{width} of the tree-decomposition $(X,T)$ is the biggest size of a bag minus one. The \emph{treewidth} of $G$ is the smallest width of $(X,T)$  amongst all the tree-decompositions $(X,T)$ of $G$.

Chordal graphs are graphs with no induced cycle of length at least $4$. A characterization given by Dirac in~\cite{dirac1961} ensures chordal graphs are graphs where minimal vertex separators are cliques. Chordal graphs admit clique trees which are tree-decompositions such that all the bags are cliques.

Our dynamic programming algorithm is performed in a bottom-up on a clique tree of the graph with more properties than the ones given by Definition~\ref{def_tree_decom}. These properties permits to simplify the analysis of the algorithm. We adapt the decomposition of~\cite[Lemma 13.1.2]{kloks1994} to get this tree-decomposition. 

\begin{lemme}\label{tree_dec_inter}
    Let $G=(V,E)$ be a chordal graph. There exists a clique tree  $(X,T)$ of $G$ such that, (i) $T$ is a rooted tree that contains at most $4n$ nodes, (ii) for every bag $i \in V(T)$, the set of vertices $X_i$ induces a clique in $G$ and (iii) $T$ contains four types of nodes which are:
\begin{itemize}
    \item Leaf nodes which satisfy $|X_i|= 1$ or,
    \item Introduce nodes $i$ which have exactly one child $j$, and that child satisfies $X_i = X_j \cup \{v\}$ for some vertex $v \in V (G) \setminus X_j$ or,
    \item Forget nodes $i$ which have exactly one child $j$, and that child satisfies $X_i = X_j \setminus \{v\}$ for some vertex $v \in X_j$ or,
    \item Join node $i$ which have exactly two children $i_1$ and $i_2$ satisfying $X_i = X_{i_1} = X_{i_2}$.
\end{itemize}
Moreover, such a clique tree can be found in linear time.
\end{lemme}

The original proof uses $k$-trees instead of chordal graphs but the proof only needs that the graph contains a simplicial vertex which holds for chordal graphs. Let us define now our clique tree in which the root of the tree is fixed.

\begin{lemme}\label{tree_dec}
    Let $G=(V,E)$ be a chordal graph and $r$ a vertex of $G$, then there exists a clique tree $(X,T)$ such that, $T$ contains at most $7n$ nodes, $T$ is rooted in a node that contains only the vertex $r$, for every bag $i \in V(T)$, the set of vertices $X_i$ induces a clique in $G$ and $T$ contains four types of nodes.
\begin{itemize}
    \item Leaf nodes, $|X_i|= 1$ which have no child.
    \item Introduce nodes $i$ which have exactly one child $j$, and that child satisfies $X_i = X_j \cup \{v\}$ for some vertex $v \in V (G) \setminus X_j$.
    \item Forget nodes $i$ which have exactly one child $j$, and that child satisfies $X_i = X_j \setminus \{v\}$ for some vertex $v \in X_j$.
    \item Join node $i$ which have exactly two children $i_1$ and $i_2$, and that children satisfy $X_i = X_{i_1} = X_{i_2}$.
\end{itemize}
Moreover, such a clique tree can be found in linear time.
\end{lemme}

\begin{proof}
By Lemma~\ref{tree_dec_inter}, there exists a clique tree  $(T,X)$ that satisfies all the properties except that the root of $T$ can be any bag.

Let us first modify $(T,X)$ to ensure that the clique tree is rooted on a node that contains only $\{r\}$.
If $r$ appears in a bag of a leaf node then it holds. Otherwise, there exists a forget node $i$ with a child $i'$ such that $X_i= X_{i'} \setminus \{r\}$. Let $X_i=\{v_1, \ldots, v_k\}$ with $r=v_1$.
We do the following modifications on $T$: delete the edge $ii'$, add two nodes $i''$ and $i_k$ such that $X_{i''}={X_k}=X_v$, $i''$ is a join node with child $i'$ and $k$ and $i$ is a forget node with child $i''$. Ultimately, we add the nodes $i_{k-1} \ldots, i_1$ such that for any $1 \leq t \leq k-1$, $X_{i_t}=\{v_1, \ldots, v_t\}$  and $i_t$ is the child of the node $i_{t+1}$ (which is a forget node). Then, $r$ appears in a bag of a leaf node by adding at most $n$ nodes in $T$. 

Let us now root $T$ on the node whose bag is $\{ r \}$. We need to check that the property on nodes are preserved. Note that for every edge, the two bags on the extremities differ at most on one vertex. If a node has only one child with the same bag then merge the two nodes. If a node $i$ had two children with different bags, let $X$ be the bag of $i$, then add a new bag with vertex set $X$ between $i$ and its child with a different bag. The tree we get after these modifications satisfies all the properties of the lemma.

All these modifications can be performed in linear time. So find the clique tree can be performed in time $O(n)$.\qed
\end{proof}

In the following,  a clique tree with the properties of Lemma~\ref{tree_dec} will be called a \emph{nice clique tree} and we will only consider nice clique trees $(X,T)$ of chordal graphs $G$. 

Given a rooted clique tree $(T,X)$ of $G$, for any node $i$ of $T$, we define the \emph{subgraph of $G$ rooted in $X_i$}, denoted by $T(X_i)$, as the subset of vertices of $G$ containing in at least one of the bags of the subtree of $T$ rooted in $i$ (i.e. in the bag of $i$ or one of its descendants).

\subsection{Clique separators and resolving sets.}\label{sec:chordal}

In this section, we give some technical lemmas that will permit to bound by $f(\omega)$ the amount of information we have to remember in the dynamic programming algorithm.

\begin{lemme}\label{resolve_far}
Let $K$ be a clique separator of $G$ and $G_1$ be a connected component of $G \setminus K$. Let $G_{ext}$ be the subgraph of $G$ induced by the vertices of $G_1 \cup K$ and $G_{int}= G \setminus G_{ext}$. Let $x_1,x_2 \in V(G_{int})$ be such that $|d(x_1,K)-d(x_2,K)| \geq 2$. Then, every vertex $s \in V(G_{ext})$ resolves the pair $(x_1,x_2)$.
\end{lemme}

\begin{proof}
Without loss of generality, assume $d(x_1,K)+2 \leq d(x_2,K)$. By triangular inequality and since $K$ is a clique, $d(x_1,s) \leq d(x_1,K) + 1 + d(K,s) $ and $d(x_2,s) \geq d(x_2,K)+d(K,s)$. The sum of these inequalities gives $d(x_2,K)+d(x_1,s) \leq d(x_1,K)+1+d(x_2,s)<d(x_2,K)+d(x_2,s)$. Thus, $d(x_1,s) <d(x_2,s)$, meaning that $s$ resolves the pair $(x_1,x_2)$.\qed
\end{proof}

Before proving Lemma~\ref{resolve_close}, let us extract a technical case.

\begin{lemme} \label{tech}
Let $G$ be a chordal graph and $T$ be a nice clique tree of $G$. Let $X$, $Y$ be two bags of $T$ and $x$, $y$ be two vertices in respectively $X$, $Y$ . Let $Y$ be a bag of $T$ such that $X \cap Y  = \emptyset$. Assume $d(x,y) \geq 2$ and let $z$ be a neighbour of $x$ that appears in the bag the closest to $Y$ in $T$ amongst the bags on the path between $X$ and $Y$. Then $z$ belongs to a shortest path between $x$ and $y$.
\end{lemme}

\begin{proof}
Let $Z$ be the bag containing $z$ and no other vertices of $N[x]$ with $Z$ on the path between $X$ and $Y$.
If $Z=Y$ then $z$ is a common neighbour of $x$ and $y$ which gives the result since $d(x,y) \geq 2$.
Otherwise, consider a shortest path $x=x_1,x_2,\ldots, x_m=y$ between $x$ and $y$ and let $x_i$ be the first vertex of this path belonging to $Z$. Such a vertex exists since $Z$ separates $x$ and $y$. If $x_i=z$ then the result holds. Otherwise by definition of $z$, $x_i$ is not adjacent to $x$ and is adjacent to $z$ because they both belong $Z$. Thus, if we replace the sub-path $x_1,\ldots,x_i$ by $x,z,x_i$, it gives a path from $x$ to $Z$ whose length is at most the length of the initial path which gives the result.\qed
\end{proof}

\begin{lemme}\label{resolve_close}
Let $S$ be a subset of vertices of $G$.  Let $X$, $Y$ and $Z$ be three bags of a nice tree-decomposition $T$ of $G$ such that $Z$ is on the path $P$ between $X$ and $Y$ in $T$. Denote by $P=X_1, \ldots Z \ldots X_p$ the bags of $P$ with $X=X_1$ and $Y=X_p$. Let $x$ be a vertex of $X$ and $y$ a vertex of $Y$ with $d(x,Z) \geq 2$ and $d(y,Z) \geq 2$.
Assume that any pair of vertices $(u,v)$ with $u \in \ X_2 \cup \ldots \cup Z$, $v \in Z \cup \ldots \cup X_p$, $d(u,Z)<d(x,Z)$ and $d(v,Z)<d(y,Z)$ is resolved by $S$. Then the pair $(u,v)$ is resolved by $S$.
\end{lemme}

\begin{proof}
Let $i_1$ be such that $X_{i_1} \cap N[x] \neq \emptyset$ and for every $j>i_1$, $X_j \cap N[x] = \emptyset$ and $i_2$ be such that $X_{i_2} \cap N[y] \neq \emptyset$ and for $j<i_2$, $X_j \cap N[y] = \emptyset$.
Let $x'$ be the only neighbour of $x$ in $X_{i_1}$ and $y'$ be the only neighbour of $y$ in $X_{i_2}$, they are unique by definition of nice tree-decomposition.
Note that $d(x,y) \geq 4$ since $d(x,Z) \geq 2$ and $d(y,Z) \geq 2$. So $N[x]$ is not adjacent to $N[y]$ and then $i_1 <i_2$.
By Lemma~\ref{tech}, $x'$ is on a shortest path between $x$ and $Z$ and $y'$ is on a shortest path between $y$ and $Z$.
So $d(x',Z)<d(x,Z)$ and $d(y',Z)<d(y,Z)$. By hypothesis, there is a vertex $s \in S$ resolving the pair $(x',y')$. Let us prove that $s$ resolves the pair $(x,y)$.

If $s$ is a neighbour of $x$ or $y$ then $s$ resolves the pair $(x,y)$ since $d(u,v) \geq 4$. So we can assume that $d(s,x) \geq 2$ and $d(s,y) \geq 2$.
Let $X_s$ be a bag of $T$ containing $s$ and $X_s'$ be the closest bag to $X_s$ on $P$ between $X$ and $Y$.

\smallskip \noindent Case 1: $s \in X_{i_1}$ and $s \in X_{i_2}$. Then, $d(s,x') \leq 1$ and  $d(s,y') \leq 1$. The vertex $s$ resolves the pair $(x',y')$ so $d(s,x') \neq d(s,y')$ so $s=x'$ or $s=y'$. Assume by symmetry that $s=x'$, then $d(s,x)=1$ and $d(s,y) \geq 3$ because $d(x,y) \geq 4$. So $s$ resolves the pair $(x,y)$.

\smallskip \noindent Case 2: $s$ belongs to exactly one of $X_{i_1}$ or $X_{i_2}$. By symmetry assume that $s \in X_{i_1}$. 
By Lemma~\ref{tech}, $y'$ is on a shortest path between $y$ and $s$. So  $d(s,y)=d(s,y')+1$.
As $s$ belongs to $X_{i_1}$ then $d(x',s)=1$ and $d(x,s) \leq 2$. As $d(y',s) \neq d(x',s)$ we have $d(y',s) \geq 2$, so $d(s,y) \geq 3$. Thus $s$ resolves the pair $(x,y)$.

\smallskip \noindent Case 3: $s \notin X_{i_1}$ and $s \notin X_{i_2}$.
First, we consider the case where $X_s'$ is between $X_{i_1}$ and $X_{i_2}$. Then, $d(s,x)=d(s,x')+1$ and $d(s,y)=d(s,y')+1$ by Lemma~\ref{tech} as $X_{i_1}$ separates $x$ and $s$ and $X_{i_2}$ separates $x$ and $s$. Thus, $s$ resolves the pair $(x,y)$.

By symmetry, we can now assume that $X_s'$ is between $X$ and $X_{i_1}$.
Since $i_1 < i_2$, $X_{i_2}$ separates $s$ and $y$. So $d(s,y)=d(s,y')+1$ by Lemma~\ref{tech}. To conclude we prove that $d(s,x') < d(s,y')$. 
Let $Q$ be a shortest path between $s$ and $y$. The bag $X_{i_1}$ separates $s$ and $y$ so $Q \cap X_{i_1} \neq \emptyset$. Let $y_1 \in Q \cap X_{i_1} $. By definition of $Q$, $d(s,y')=d(s,y_1)+d(y_1,y)$. We know $y_1 \neq y$ because $y_1$ is a neighbour of $x$. So $d(y_1,y) \neq 0$. We also have $d(s,x')\leq d(s,y_1)+1$ because $y_1 \in X_{i_1}$. So $y_1$ is a neighbour of $x'$. As $d(s,x')\neq d(s,y')$, this ensures $d(s,x') < d(s,y')$. So $s$ resolves the pair $(x,y)$ because $d(s,x) \leq d(s,x')+1 < d(s,y')+1=d(s,y)$.\qed
\end{proof}

The following corollary is essentially rephrasing Lemma~\ref{resolve_close} to get the result on a set of vertices.

\begin{cor}\label{resolve_dec}
Let $G$ be a chordal graph and $S$ be a subset of vertices of $G$. Let $X_i$ be a bag of $T$ and let $T_1=(X_1,E_1)$ and $T_2=(X_2,E_2)$ be two connected components of $T \setminus X_i$. Assume that any pair of vertices $(u,v)$ of $(X_1 \cup X_i) \times (X_2 \cup X_i)$ with $d(u,X_i)\leq 2$ and $d(v,X_i) \leq 2$  is resolved by $S$. Then any pair of vertices $(u,v)$ of $(X_1,X_2)$ with $|d(u,X_i)-d(v,X_i)| \leq 1$ is resolved by $S$.
\end{cor}

\begin{proof}
Assume by contradiction, that there exist some pairs of vertices of $(T_1,T_2)$ with $|d(u,X_i)-d(v,X_i)| \leq 1$ which are not resolved by $S$. Among all these pairs, let $(u,v)$ be one pair minimizing $q:=d(u,X_i)+d(v,X_i)$.
If $q \leq 4$ then $d(u,X_i)\leq 2$ and $d(v,X_i)\leq 2$, so the pair $(u,v)$ is resolved by the hypothesis of the lemma.
If $q \geq 5$, then $d(u,X_i)\geq 2$ and $d(v,X_i)\geq 2$. By minimality, we know that all pairs $(u',v')$ of $(T_1,T_2)$ with $d(u',X_i)<d(u,X_i)$ and $d(v',X_i)<d(v,X_i)$ are resolved by $S$. So, by Lemma~\ref{resolve_close}, the pair $(u,v)$ is resolved by $S$.\qed
\end{proof}

\section{Algorithm description}

In this section, we fix a vertex $v$ of a chordal graph $G$ and consider a nice clique tree $(T,X)$ rooted in $v$ which exists by Lemma~\ref{tree_dec}. We present an algorithm computing the smallest size of a resolving set of $G$ containing $v$.

\subsection{Generalisation of the problem}\label{sec:defpb}

The algorithm is a dynamic programming algorithm that works bottom-up from the leaves of a nice clique tree. Our algorithm computes the solution of a problem more general than the metric dimension but easiest to manipulate for combining solutions. Our algorithm consists in a dynamic programming on the clique tree. In this new problem, we will represent some vertices by vectors of distance.

We define notations to edit vectors.
\begin{defn} 
Given a vector $\vec{r}$, the notation $\vec{r}_i$ refers to the $i$-th coordinate of~$\vec{r}$.
\begin{itemize}
    \item Let $\vec{r}=(r_1,\ldots,r_k) \in \mathbb{N}^k$ be a vector of size $k$ and $m \in N$. The vector $\vec{r'}=\vec{r|m}$ is the vector of size $k+1$ with $r'_i=r_i$ for $1 \leq i \leq k$ and $r'_{k+1}=m$. 
    \item Let $\vec{r}=(r_1,\ldots,r_k) \in \mathbb{N}^k$ be a vector of size $k$. The vector $\vec{r^-}$ is the vector of size $k-1$ with $r^-_i=r_i$ for $1 \leq i \leq k-1$. 

\end{itemize}
\end{defn}

\begin{defn}
Let $i$ be a node of $T$ and let $X_i=\{v_1,\ldots,v_k\}$ be the bag of $i$. For a vertex $x$ of $G$, the \emph{distance vector} $\distvec{X_i}{x}$ of $x$ to $X_i$ is the vector of size $k$ such that, for $1 \leq j \leq k$, $\distvec{X_i}{x}_j=d(x,v_j)$.
We define the set $\ddeux{X_i}$ as the set of distance vectors of the vertices of $T(X_i)$ at distance at most $2$ of $X_i$ in $G$ (i.e. one of the coordinate is at most $2$). 
\end{defn}

\begin{defn}
Let $G$ be a graph and $K=\{v_1,\ldots,v_k\}$ be a clique of $G$.
Let $x$ be a vertex of $G$. The \emph{trace} of $x$ on $K$, denoted by $\vec{Tr_K}(x)$, is the vector $\vec{r}$ of $\{0,1\}^k\setminus \{1, \ldots ,1\}$ such that for every $1 \leq i \leq k$, $d(x,v_i)=a+\vec{r}_i$ where $a=d(x,K)$.

Let $S$ be a subset of vertices of $G$. The trace $Tr_K(S)$ of $S$ in $K$ is the set of vectors $ \{ \vec{Tr_K}(x), {x\in S} \}$.
\end{defn}

The trace is well-defined because for a vertex $x$ and a clique $K$, the distance between $x$ and a vertex of $K$ is either $d(x,K)$ or $d(x,K)+1$.

\begin{defn}\label{def_resolve_vec}
Let $\vec{r_1}, \vec{r_2}$ and $\vec{r_3}$ be three vectors of same size $k$. We say that $\vec{r_3}$ \emph{resolves} the pair $(\vec{r_1},\vec{r_2})$ if \[\min_{1 \leq i \leq k} \vec{(r_1+r_3)}_i \neq \min_{1 \leq i \leq k} \vec{(r_2+r_3)}_i.\]
\end{defn}

\begin{lemme}
Let $K$ be a clique separator of $G$ and $G_1$ be a connected component of $G \setminus K$. Let $(x,y)$ be a pair of vertices of $G \setminus G_1$ and let $\vec{r}$ be a vector of size $|K|$. If $\vec{r}$ resolves the pair $(\distvec{K}{x},\distvec{K}{y})$, then any vertex $s \in V(G_1)$ with $\vec{Tr_K}(s)=\vec{r}$ resolves the pair $(x,y)$.
\end{lemme}

\begin{proof}
Let $s$ be a vertex of $G_1$ such that $\vec{Tr_K}(s)=\vec{r}$. The clique $K$ separates $s$ and $x$ (resp. $y$) so $d(x,s)=\min_{1 \leq i \leq k} (\distvec{K}{x}+\vec{Tr_K}(s))_i+d(K,s)$ (resp. $d(y,s)=\min_{1 \leq i \leq k} (\distvec{K}{y}+\vec{Tr_K}(s))_i+d(K,s)$). The vector $\vec{r}$ resolves the pair $(\distvec{K}{x},\distvec{K}{y})$. So $d(x,s) \neq d(y,s)$ and $s$ resolves the pair $(x,y)$.\qed
\end{proof}

\begin{defn}
Let $K$ be a clique separator of $G$ and $G_1$, $G_2$ be two (non necessarily distinct) connected components of $G \setminus K$. Let $M$ be a set of vectors and let $u \in V(G_1) \cup K$ and $v \in V(G_2) \cup K$. If a vector $\vec{r}$  resolves the pair $(\distvec{K}{x},\distvec{K}{y})$, we say that $\vec{r}$ \emph{resolves} the pair $(x,y)$.
We say that the pair of vertices  $(u,v)$ \emph{is resolved} by $M$ if there exists a vector $\vec{r} \in M$ that resolves the pair $(u,v)$.
\end{defn}

We can now define the generalised problem our dynamic programming algorithm actually solves. We call it the \textsc{extended metric dimension} problem ({\sc EMD} for short) .
We first define the instances of this problem.

\begin{defn}Let $i$ be a node of $T$.
An \emph{instance for a node $i$} of the {\sc EMD} problem is a $5$-uplet $\I {}$ composed of the bag $X_i$ of $i$, a subset $S_I$ of $X_i$ and three sets of vectors satisfying

\begin{itemize}
    \item $\int {I}\subseteq \{0,1\}^{|X_i|}$ and $\ext{I} \subseteq \{0,1\}^{|X_i|}$,
    \item $\paire{I} \subseteq [|0,3|]^{|X_i|} \times [|0,3|]^{|X_i|}$,
    \item $\ext{I} \neq \emptyset$ or $S_I \neq \emptyset$,
    \item For each pair of vectors $(\vec{r_1},\vec{r_2}) \in \paire I$, there exist two vertices $x \in \T {X_i}$ with $\distvec{X_i}{x}= \vec{r_1}$ and $d(x,X_i) \leq 2$ and $y \notin \T {X_i}$ with $\distvec{X_i}{y}= \vec{r_2}$ and $d(y,X_i) \leq 2$.
\end{itemize} 

\end{defn}

\begin{defn}\label{def_instance}
A set $S \subseteq T(X_i)$ is a solution for an instance $I$ of the {\sc EMD} problem if 
\begin{itemize}
    \item \textbf{(S1)} Every pair of vertices of $\T {X_i}$ is  either resolved by a vertex in $S$ or resolved by a vector of $\ext I$. 
    \item \textbf{(S2)} For each vector $\vec{r} \in \int I$ there exists a vertex $s \in S$ such that $\vec{Tr_{X_i}}(s)= \vec{r}$. 
    \item \textbf{(S3)} For each pair of vector $(\vec{r_1},\vec{r_2}) \in \paire I$, for any vertex $x \in \T {X_i}$ with $\distvec{X_i}{x}= \vec{r_1}$ and any vertex $y \notin \T {X_i}$ with $\distvec{X_i}{y}= \vec{r_2}$, if $d(x,X_i) \leq 2$ and $d(y,X_i) \leq 2$ the pair $(x,y)$ is resolved by $S$. 
    \item \textbf{(S4)} $S \cap X_i =S_I$. 
\end{itemize}

\end{defn}
In the rest of the paper, for shortness, we will refer to an instance of the {\sc EMD} problem only by an instance.

\begin{defn}\label{def_dim}
Let $I$ be an instance. We denote by $\dim (I)$ the minimum size of a set $S \subseteq \T{X_i}$ which is a solution of $I$. If such a set does not exist we define $\dimx I= + \infty$. We call this value the \emph{extended metric dimension} of $I$.
\end{defn}

We now explain the meaning of each element of $I$. Firstly, a solution $S$ must resolve any pair in $\T {X_i}$, possibly with a vector of $\ext{I}$ which represents a vertex of $V \setminus \T {X_i}$ in the resolving set. Secondly, for all $\vec{r}$ in $\int I$, we are forced to select a vertex in $T(X_i)$ whose trace is $\vec{r}$. This will be useful to combine solutions since it will be a vector of $D_{ext}$ in other instances. The elements in $\paire{I}$ will also be useful for combinations. In some sense $\paire{I}$ is the additional gain of $S$ compared to the main goal to resolve $\T {X_i}$. The set $S_I$ constrains the intersection between $S$ and $X_i$ by forcing a precise subset of $X_i$ to be in $S$.

The following lemma is a consequence of Definition~\ref{def_instance}. It connects the definition of the extended metric dimension with the metric dimension.
\begin{lemme}\label{correct}
Let $G$ be a graph, $T$ be a nice tree-decomposition of $G$ and $r$ be the root of $T$. Let $I_0$ be the instance ${(\{r\},\{r\},\emptyset,\emptyset,\emptyset)}$, then $\dim(I_0)$ is the smallest size of a resolving set of $G$ containing $r$.
\end{lemme}

To ensure that our algorithm works well, we will need to use Lemma~\ref{resolve_far} in some subgraphs of $G$. This is possible only if we know that the solution is not included in the subgraph. This corresponds to the condition $\ext{I} \neq \emptyset$ or $S_I \neq \emptyset$ and this is why the algorithm computes the size of a resolving set containing the root of $T$.

\subsection{Dynamic programming} \label{sec:rules}

We explain how we can compute the metric dimension of an instance $I$  given the metric dimension of the instances on the children of $X_i$ in $T$. The proof is divided according to the different type of nodes.

\subsubsection{Leaf node}

Computing the dimension of an instance for a leaf node can be done easily with the following lemma.
\begin{lemme}\label{calcul_leaf}
Let $I$ be an instance for a leaf node $i$ and $v$ be the unique vertex of~$X_i$. Then, 

$$
\dim(I) = \left\{
    \begin{array}{lll}
        0 & \text{ if } S_I=\emptyset,\; \int I =\emptyset \text{ and } \paire I = \emptyset \\
        1 & \text{ if } S_I=\{v\} \text{ and } \int {I} \subseteq \{ \vec{(0)} \} \\
        +\infty & \text{ otherwise}
    \end{array}
\right.
$$
\end{lemme}

\begin{proof}
Let $I$ be an instance for $i$. If $S_I=\emptyset$, only the set $S=\emptyset$ can be a solution for $I$. This set is a solution only if $\int I =\emptyset$ and $\paire I = \emptyset$.
If $S_I=\{v\}$, only the set $S=\{v\}$ can be a solution for $I$. This is a solution only if $\int {I}$ is empty or only contains the vector $\vec{Tr_{x_i}}(v)$. \qed
\end{proof}

In the rest of the section, we treat the three other types of nodes.
For each type of nodes we will proceed as follows: define some conditions on the instances on children to be compatible with $I$, and prove an equality between the extended metric dimension on compatible children instances and the extended metric dimension of the instance of the node. 

\subsubsection{Join node.} Let $I$ be an instance for a join node $i$ and let $i_1$ and $i_2$ be the children of $i$.

\begin{defn}\label{compatible_pair}
A pair of instances $(I_1, I_2)$ for $(i_1,i_2)$ is \emph{compatible} with~$I$ if 
\begin{itemize}
    \item \textbf{(J1)} $S_{I_1}=S_{I_2}=S_I$,
    \item \textbf{(J2)} $ \ext {I_1} \subseteq \ext{I} \cup \int{I_2} $ and $ \ext {I_2} \subseteq \ext{I} \cup \int{I_1} $,
    \item \textbf{(J3)} $\int{I} \subseteq \int{I_1} \cup \int {I_2}$,

    \item \textbf{(J4)} 
    Let $C_1=\{ (\vec{r},\vec{t}) \in \paire{I}$ such that $\vec{r} \notin \ddeux{X_{i_1}}\}$ and $C_2=\{ (\vec{r},\vec{t}) \in \paire {I}$ such that $\vec{r} \notin \ddeux{X_{i_2}}\}$.
    Let $D_1=\{ (\vec{r},\vec{t}) \in \ddeux{X_{i_1}} \times \ddeux{G \setminus X_{i_1}}$ such that there exists $\vec{u} \in \int{I_2}$ resolving the pair $(\vec{r},\vec{t}) \}$ and $D_2=\{ (\vec{r},\vec{t}) \in \ddeux{X_{i_2}} \times \ddeux{G \setminus X_{i_2}})$ such that there exists $\vec{u} \in \int{I_1}$ resolving the pair $(\vec{r},\vec{t})\}$

    Then $ \paire{I} \subseteq (\vec{r},\vec{t}) \in (C_1 \cup D_1 \cup \paire{I_1}) \cap (C_2 \cup D_2 \cup \paire{I_2}) $,

    \item \textbf{(J5)} For all $ \vec{r_1} \in \ddeux{X_{i_1}}$, for all $\vec{r_2} \in \ddeux{X_{i_2}}$, $(\vec{r_1},\vec{r_2}) \in \paire {I_1}$ or $(\vec{r_2},\vec{r_1}) \in \paire {I_2}$ or there exists $\vec{t} \in \ext{I}$ such that $\vec{t}$ resolves the pair $(\vec{r_1},\vec{r_2})$.
    
\end{itemize}
\end{defn}  

Condition \textbf{(J4)} represents how the pairs of vertices of $V(T(X_{i_1})) \times V(T(X_{i_2}))$ can be resolved. A pair $(\vec{r},\vec{t})$ is in  $(C_1 \cup D_1 \cup \paire{I_1})$ if all the pairs of vertices $(x,y)$ with $x \in V(T(X_{i_1})) $ and $y \in V(T(X_{i_2}))$ are resolved. If $(\vec{r},\vec{t})$ is in $C_1$, no such pair of vertices exists, if $(\vec{r},\vec{t})$ is in $D_1$ the pairs of vertices are resolved by a vertex outside of $V(T(X_{i_1}))$ and if $(\vec{r},\vec{t})$ is in $\paire{I_1}$ the pairs of vertices are resolved by a vertex of $V(T(X_{i_1}))$. So a pair $(\vec{r},\vec{t})$ is resolved if the pair is in $(C_1 \cup D_1 \cup \paire{I_1})$ and in $(C_2 \cup D_2 \cup \paire{I_2})$.

Let $\mathcal{F}_J(I)$ be the set of pairs of instances compatible with~$I$.
We want to prove the following lemma:

\begin{lemme}\label{node_pair_main}
Let $I $ be an instance for a join node $i$. Then, 
\[ \dim(I) = \min_{(I_1,I_2) \in \mathcal{F}_J(I)} (\dim(I_1) + \dim(I_2) - |S_I|). \]
\end{lemme}
 We prove the equality by proving the two inequalities in the next lemmas.

\begin{lemme}\label{calcul_pair_ineq1}
 Let $(I_1,I_2)$ be a pair of instances for $(i_1,i_2)$ compatible with $I$ with finite values for $\dim(I_1)$ and $\dim(I_2)$. Let $S_1 \subseteq V(T(X_{i_1}))$ be a solution for $I_1$ and $S_2 \subseteq V(T(X_{i_2}))$ be a solution for $I_2$. Then $S = S_1 \cup S_2$ is a solution for $I$. In particular, $$  \dim(I) \leq \min_{(I_1,I_2) \in \mathcal{F}_J(I)} (\dim(I_1) + \dim(I_2) - |S_I|).$$ 
\end{lemme}

\begin{proof}
Let us prove that the conditions of Definition~\ref{def_instance} are satisfied.

\noindent \textbf{(S1)} Let $(x,y)$ be a pair of vertices of $T(X_i)$. Assume first that $x \in V(T(X_{i_1}))$ and $y \in V(T(X_{i_1}))$. Either $(x,y)$ is resolved by a vertex of $S_1$ and then by a vertex of $S$ or $(x,y)$ is resolved by a vector $\vec{r} \in \ext{I_1}$.
By condition \textbf{(J2)}, $\vec{r} \in \ext{I}$ or $\vec{r} \in \int{I_2}$. If $\vec{r} \in \ext{I}$ then $(x,y)$ is resolved by a vector of $\ext{I_1}$. Otherwise, there exists a vertex $t \in S_2$ such that $\vec{Tr_{X_{i_2}}}(t)=\vec{r}$. So $t \in S$ and $t$ resolves the pair $(x,y)$. The case $x \in V(T(X_{i_2}))$ and $y \in V(T(X_{i_2}))$ is symmetric. So we can assume that $x \in V(T(X_{i_1}))$ and $y \in V(T(X_{i_2}))$. If $d(x,X_i) \leq 2$ and $d(y,X_i) \leq 2$, the condition \textbf{(J5)} ensures that the pair $(x,y)$ is resolved by $S$ or by a vector of $\ext{I}$.
Otherwise, either $|d(x,X_i) - d(y,X_i)| \leq 1$ and $(x,y)$ is resolved by Lemma~\ref{resolve_dec} or $|d(x,X_i) - d(y,X_i)| \geq 2$ and $(x,y)$ is resolved by Lemma~\ref{resolve_far} because $\ext{I} \neq \emptyset$ or $  S_I \neq \emptyset$.

\noindent \textbf{(S2)} Let $\vec{r} \in \int{I}$. By compatibility, the condition \textbf{(J3)} ensures that $\vec{r} \in \int{I_1}$ or $\vec{r} \in \int{I_2}$. As $S=S_1 \cup S_2$, $S$ contains a vertex $s$ such that $\vec{Tr_{X_i}}(s)=\vec{r}$.

\noindent \textbf{(S3)} Let $(\vec{r},\vec{t}) \in \paire{I}$ and $(x,y)$ with $x \in V(T(X_i))$ such that $\distvec{X_i}{x}=\vec{y}$ and  $y \notin T(X_i)$ such that $\distvec{X_i}{y}=\vec{t}$. Without loss of generality assume that $x \in V(T(X_{i_1}))$.

By compatibility, $ (\vec{r},\vec{t}) \in (C_1 \cup D_1 \cup \paire{I_1}) \cap (C_2 \cup D_2 \cup \paire{I_2})$ so in $C_1 \cup D_1 \cup \paire{I_1}$.
If $(\vec{r},\vec{t}) \in \paire I_1$, then there exists $s \in S_1$ that resolves the pair $(x,y)$ so the pair is resolved by $S$. If $(\vec{r},\vec{t}) \in D_1$, there exists $\vec{u} \in \int {I_2}$ such that $\vec{u}$ resolves the pair $(\vec{r},\vec{t})$. By compatibility, there exists $s \in S_2$ such that $\vec{Tr_{X_i}}(s)=\vec{u}$. So $s$ resolves the pair $(x,y)$.
And $(\vec{r},\vec{t}) \notin  C_1$ since $x$ belongs to $T(X_{i_1})$ with vector distance $\vec{r}$.

\noindent \textbf{(S4)} is clear since $X_{i_1}=X_{i_2}=X_{i}$.

Thus, $\dim(I) \leq \dim(I_1)+\dim(I_2)-|S_I|$ is true for any pair of compatible instances $(I_1,I_2)$ so $  \dim(I) \leq \min_{(I_1,I_2) \in \mathcal{F}_J(I)} (\dim(I_1) + \dim(I_2) - |S_I|)$.\qed
\end{proof}
\begin{lemme}\label{calcul_pair_ineq2}
Let $I $ be an instance for a join node $i$ and let $i_1$ and $i_2$ be the children of $i$. Then, 
 \[ \dim(I) \geq \min_{(I_1,I_2) \in \mathcal{F}_J(I)} (\dim(I_1) + \dim(I_2) - |S_I|). \]
\end{lemme}

\begin{proof}
If $\dim (I)=+ \infty$ then the result indeed holds. So assume $\dim (I)$ is finite.
Let $S$ be a solution for $I$ of minimal size. Let $S_1=S \cap T(X_{i_1})$ and $S_2=S \cap T(X_{i_2})$. We define now two instances $I_1$ and $I_2$ for $i_1$ and $i_2$. Let $S_{I_1} = S_{I_2}=S_I$, $\int {I_1}= Tr_{X_i}(S_1)$, $\int {I_2}= Tr_{X_i}(S_2)$, $\ext{I_1}=\ext{I} \cup \int{I_2}$ and $\ext{I_2}=\ext{I} \cup \int{I_1}$. To build the sets $\paire{I_1}$ and $\paire{I_2}$ we make the following process that we explain for $\paire{I_1}$. For all pairs of vectors $(\vec{r},\vec{t})$ of $(\ddeux{X_{i_1}},\ddeux{G \setminus X_{i_1}})$, consider all the pairs of vertices $(x,y)$ with $x \in V(T(X_{i_1}))$, $y \in V(G \setminus T(X_{i_1}))$, $\vec{r} \in \ddeux{X_i}$, $\vec{t} \in \ddeux{G \setminus X_{i_1}})$, $\distvec{X_i}{x}=\vec{r}$ and $\distvec{X_i}{y}=\vec{t}$. If all the pairs are resolved by vertices of $S_1$ (that for each pair, there exists a vertex of $S_1$ that resolves the pair), then add $(\vec{r},\vec{t})$ to $\paire{I_1}$.
\begin{claim}
$(I_1,I_2)$ is compatible with $I$.

\end{claim}
\begin{proof}
    
\textbf{(J1)}, \textbf{(J2)} and \textbf{(J3)} are straightforward.

\noindent \textbf{(J4)} Let $(\vec{r},\vec{t}) \in \paire I$, we want to prove that $ (\vec{r},\vec{t}) \in (C_1 \cup D_1 \cup \paire{I_1}) \cap (C_2 \cup D_2 \cup \paire{I_2}) $. We prove that $  (\vec{r},\vec{t}) \in (C_1 \cup D_1 \cup \paire{I_1})$, the other part of the proof is symmetrical.

If $\vec{r} \notin \ddeux{X_{i_1}}$, then $ (\vec{r},\vec{t}) \in C_1 $. Otherwise, there exists a vertex $x$ in $T(X_{i_1})$ such that $\distvec{X_{i_1}}{x}=\vec{r}$ and a vertex $y$ in $G \setminus T(X_i)$ such that $\distvec{X_i}{y}=\vec{t}$ (because the pair $(\vec{r},\vec{t})$ belongs to $\paire I$). The pair $(x,y)$ is resolved by $S$. If there is a vertex $s \in S \cap T(X_{i_2})$ resolving the pair, then $s$ resolves all the pairs with such distance vector and then $(\vec{r},\vec{t}) \in D_1$. Otherwise, for any pair $(x,y)$ of $ T(X_{i_1}) \times G \setminus T(X_i) $ with $\distvec{X_{i_1}}{x}=\vec{r}$ and $\distvec{X_i}{y}=\vec{t}$, there is a vertex of  $S \cap T(X_{i_1})$ that resolves the pair $(x,y)$, so $(\vec{r},\vec{t}) \in \paire {I_1}$.

\noindent \textbf{(J5)} Let $\vec{r_1} \in \ddeux{X_{i_1}}$, $\vec{r_2} \in \ddeux{X_{i_2}}$ and two vertices $x \in X_{i_1}$ and $y \in X_{i_2}$ such that $\distvec{X_{i_1}}{x}=\vec{r_1}$ and $\distvec{X_{i_2}}{y}=\vec{r_2}$. As $S$ is a solution of $I$, either the pair $(x,y)$ is resolved by a vector $\vec{r_3} \in \ext I$, or there exists $s \in S$ resolving $(x,y)$. If $(x,y)$ is resolved by $s$, assume by contradiction that $(\vec{r_1},\vec{r_2}) \not \in \paire{I_1}$ and $(\vec{r_2},\vec{r_1}) \not \in \paire{I_1}$. Then there exist vertices $x_1,x_1' \in V(T(X_{i_1}))$ with $\distvec{X_{i_1}}{x_1}=\vec{r_1}$, $\distvec{X_{i_2}}{x_1'}=\vec{r_1}$ and $x_2,x_2' \in V(T(X_{i_2}))$ with $\distvec{X_{i_2}}{x_2}=\vec{r_2}$ and $\distvec{X_{i_2}}{x_2'}=\vec{r_2}$ such that the pair $(x_1,x_2)$ is not resolved by a vertex of $S_1$ and the pair $(x_1',x_2')$ is not resolved by a vertex of $S_2$. Let $s \in S$ resolving the pair $(x_1,x_2')$. If $s \in S_1$, then $s$ resolves the pair $(x_1,x_2)$ and if $s \in S_2$, then $s$ resolves the pair $(x_1',x_2')$, a contradiction. \qed 
\end{proof}
\begin{claim}
$S_1$ is a solution of $I_1$ and $S_2$ is a solution of $I_2$.

\end{claim} 
\begin{proof}

We only prove that $S_1$ is a solution of $I_1$ as the proof that $S_2$ is a solution of $I_2$ is similar.
\smallskip \noindent

\noindent \textbf{(S1)} Let $(x,y)$ be a pair of vertices of $T(X_{i_1})$. As $S$ is a solution of $I$, the pair $(x,y)$ is either resolved by a vertex of $S$ or by a vector of $\ext{I}$. If $(x,y)$ is resolved by a vector of $\ext{I}$, the pair $(x,y)$ is also resolved by a vector of $\ext{I_1}$ since $\ext{I} \subseteq \ext{I_1}$. Otherwise let $s \in S$ resolving the pair $(x,y)$. If $s \in T(X_{i_1})$ then $(x,y)$ is resolved by a vertex of $S_1$. Otherwise $s \in T(X_{i_2})$ and by construction of $I_1$, $\ext {I_1}$ contains the vector $\vec{Tr_{X_i}(s)}$ so $(x,y)$ is resolved by a vector of $\ext{I_1}$.

\noindent \textbf{(S2)} By definition, $ \int{I_1}=Tr_{X_i}(S_1)$. Hence, for any vector $\vec{r} \in \int {I_1}$, there is a vertex $s \in S_1$ with $\vec{Tr_{X_i}(s)}=\vec{r}$.

\noindent \textbf{(S3)} Let $(\vec{r},\vec{t}) \in \paire{I_1}$, $x \in T(X_{i_1})$ and $y\notin T(X_{i_1})$ such that $\distvec{X_{i}}{x}=\vec{r}$ and $\distvec{X_{i}}{y}=\vec{y}$. By construction of $\paire{I_1}$ there is a vertex $s \in S_1$ resolving the pair $(x,y)$.

\noindent \textbf{(S4)} $S_{I_1}=S_I$ and since $S$ is a solution of $I$, $S_I=S \cap X_i$. \qed
\end{proof}
Ultimately we get the announced inequality. Since $S$ is a minimal solution for $I$, we have $\dim(I)=|S|$. The sets $S_1$ and $S_2$ are solutions for $S_1$ and $S_2$ so $\dim(I_1) \leq |S_1|$ and $\dim(I_2) \leq |S_2|$. Since $|S|=|S_1|+|S_2|-|S_I|$ we get $\dim(I) \geq \dim(I_1) +\dim(I_2)- |S_I|$. This inequality is true for a specific pair of instances so in particular is true for a pair minimising the amount $\dim(I_1) +\dim(I_2)- |S_I|$, giving the result. \qed
\end{proof}

Lemma~\ref{node_pair_main} is a direct consequence of Lemma~\ref{calcul_pair_ineq1} and Lemma~\ref{calcul_pair_ineq2}.

\subsubsection{Introduce node}

We now consider an instance $I$ for an introduce node $i$. Let $j$ be the child of $i$ and $v \in V$ be such that $X_i=X_j \cup \{v\}$. Let $X_i=\{v_1,\ldots,v_k\}$ with $v=v_k$.
The tree $T(X_i)$ contains one more vertex than its child. The definition of the compatibility is slightly different if we consider the same set as a solution (type $1$) or if we add this vertex to the resolving set (type $2$).

\begin{defn}\label{compatible_introduce}
An instance $I_1$ is compatible with $I$ of type $1$ (resp. $2$) if
\begin{itemize}
    \item \textbf{(I1)} $S_I=S_{I_1}$ (resp. $=S_{I_1} \cup \{v\}$). 
    \item \textbf{(I2)} For all $  \vec{r} \in \ext{I}$,  $\vec{r^-} \in \ext{I_1}$ (resp. or $\vec{r}=(0,\ldots,0)$).
    \item \textbf{(I3)} For all $\vec{r} \in \int{I}, \vec{r}_k=1$ and  $\vec{r^-}\in \int{I_1}$ (resp. or $\vec{r}=(1,\ldots,1,0)$).
    \item \textbf{(I4)} For all $ (\vec{r},\vec{t}) \in \paire {I}$, $(\vec{r^-},\vec{t^-}) \in \paire {I_1}$.
    \item \textbf{(I5)} If $I_1$ is of type $1$, for all $(\vec{r},\vec{t})$ with $\vec{t}=(0,\ldots,0)$, $(\vec{r},\vec{t}) \in \paire{I_1}$.
    
\end{itemize}
\end{defn}

We want to prove that the following holds:
\begin{lemme}\label{node_introduce_main}
    
Let $I$ be an instance for an introduce node $i$. 
Let $\mathcal{F}_{1}(I)$ be the set of instances $I_1$ for $i_1$ compatible with $I$ of type 1 and $\mathcal{F}_{2}(I)$ be the set of instances $I_2$ for $i_1$ compatible with $I$ of type~2. Then,

\[\dim(I)=\min\;\{ \min_{I_1 \in \mathcal{F}_{1}(I)}\;\{\dim(I_1)\};\min_{I_2 \in \mathcal{F}_{2}(I)}\;\{\dim(I_2)+1\}\}.\]
\end{lemme}

Let us first prove a technical case.
\begin{lemme}\label{resolved_outside}
Let $i$ be an introduce node, $j$ be the child of $i$ and $v \in V$ such that $X_i=X_j \cup \{v\}$.
Let $(x,y)$ be a pair of vertices of $ T(X_j)$. Let $\vec{r}$ be a binary vector of size $|X_i|$, then $\vec{r}$ resolves $(x,y)$ if and only if $\vec{r^-}$ resolves $(x,y)$.
\end{lemme}

\begin{proof}
Let $\vec{r_1}=\vec{Tr_{X_i}}(x)$ and $\vec{r_2}=\vec{Tr_{X_i}}(y)$. Note that the set ${X_i\setminus \{v\}}$ separates $v$ from $x$ and $y$ so $\vec{(r_1)}_k=\vec{(r_2)}_k=1$.

Assume first that $\vec{r}$ resolves $(x,y)$ and by contradiction that $\vec{r^-}$ does not resolve $(x,y)$. Since $\vec{r}$ resolves $(x,y)$,
$\min_{1 \leq l \leq k} \vec{(r_1+r)}_\ell \neq \min_{1 \leq \ell \leq k} \vec{(r_2+r)}_\ell$ and $\min_{1 \leq \ell \leq k-1} \vec{(r_1+r)}_\ell = \min_{1 \leq \ell \leq k-1} \vec{(r_2+r)}_\ell$ by Definition~\ref{def_resolve_vec}.
So the minimum change in at least one case. Assume by symmetry that $\min_{1 \leq \ell \leq k} \vec{(r_1+r)}_\ell \neq \min_{1 \leq \ell \leq k-1} \vec{(r_1+r)}_\ell$. So for $\ell<k$ we have $\vec{(r_1+r)}_\ell > \vec{(r_1+r)}_k$. Since $(\vec{r_1)}_k=1$, it implies that $d(x,v)<d(x,v_j)$  for all $
\ell \leq k-1$. A contradiction since $\{v_1, \ldots, v_{k-1}\}$ separates $x$ from $v$.

Assume now that $\vec{r^-}$ resolves $(x,y)$ and by contradiction that $\vec{r}$ does not resolve $(x,y)$. Then by symmetry we can assume that $\min_{1 \leq j \leq k} \vec{(r_1+r)}_j \neq \min_{1 \leq j \leq k-1} \vec{(r_1+r)}_j$ meaning $ \vec{(r_1+r)}_k < \min_{1 \leq j \leq k-1} \vec{(r_1+r)}_j$. Since $(\vec{r_1)}_k=1$, $\vec{(r_1+r)}_k \geq 1$ and $\vec{(r_1+r)}_j =2$ for $1 \leq j \leq k-1$. So $\vec{r_1}=(1,\ldots1)$ which contradicts the fact that ${X_i\setminus \{v\}}$ separates $v$ from $x$.\qed
\end{proof}

\begin{lemme} \label{join_node1}
Let $I_1$ be a compatible instance of type $1$ and $S$ be a solution of $I_1$, then $S$ is a solution of $I$.
\end{lemme}
\begin{proof}
Let us prove that the conditions of Definition~\ref{def_instance} are satisfied.

\noindent \textbf{(S1)} Let $(x,y)$ be a pair of vertices of $T(X_i)$. First assume that  $x \neq v$ and $y \neq v$. If the pair $(x,y)$ is not resolved by a vertex of $S$, since $S$ is a solution for $I_1$, $(x,y)$ is resolved by a vector of $\ext I$. Let $\vec{r} \in \ext I$ resolving the pair $(x,y)$. As $I_1$ is compatible of type $1$, $\vec{r^-} \in \ext{I_1}$. Then $\vec{r^-}$ resolves $(x,y)$ by Lemma~\ref{resolved_outside}. So we can assume that $x=v$. The pair $(x,y)$ is also resolved by $S$ since $(\distvec{X_j}{x_1} ,(0,\ldots,0)) \in \paire{I_1}$. As $S$ is a solution for $I_1$, there is a vertex $s \in S$ that resolves the pair $(x,y)$.

\noindent \textbf{(S2)} Let $\vec{r} \in \int {I}$. Since $I_1$ is compatible with $I$, there exists $\vec{r_1}\in \int{I_1}$ such that $\vec{r}=\vec{r_1|1}$. Let $s \in S$ such that $\vec{Tr_{X_j}}(s)=\vec{r_1}$, then $\vec{Tr_{X_i}}(s)=\vec{r}$. Indeed, the vertex $v$ is not the closest vertex of $X_i$ from $s$. If $s \in X_j$, that is clear. Otherwise $X_j$ is a separator between $s$ and $v$, so the shortest path between $s$ and $v$ crosses a vertex of $X_j$. Thus, $\vec{Tr_{X_i}}(s)=\vec{r_1|1}$.

\noindent \textbf{(S3)} Let $(\vec{r},\vec{t}) \in \paire I$. Let $x \in V(T(X_i))$ such that $\distvec{X_i}{x}=\vec{r}$ and $y \notin T(X_i)$ such that $\distvec{X_i}{y}=\vec{t}$. Then  $\distvec{X_j}{x_1}=\vec{r^-}$ and $\distvec{X_j}{x_2}=\vec{t^-}$ so the pair $(x,v)$ is resolved by $S$ because $(\vec{r_1^-},\vec{r_2^-})$ belongs to $\paire{I_1}$.

\noindent \textbf{(S4)} As $S_I=S_{I_1}$ we have $S \cap X_i=S_I$.\qed
\end{proof}
\begin{lemme}\label{join_node2}
Let $I_2$ be a compatible instance of type $2$ and $S$ a solution of $I_2$, then $S'=S \cup \{v\}$ is a solution of~$I$.
\end{lemme}
\begin{proof}

Let us prove that the conditions of Definition~\ref{def_instance} are satisfied.

\noindent \textbf{(S1)} Let $(x,y)$ be a pair of vertices of $T(X_i)$ with $x \neq v$ and $y \neq v$. If the pair $(x,y)$ is not resolved by a vertex of $S$, since $S$ is a solution for $I_2$, $(x,y)$ is resolved by a vector $\vec{r} \in \ext {I_2}$. By compatibility there exists $\vec{r'} \in \ext I$ such that $\vec{r'^-}=\vec{r}$. By Lemma~\ref{resolved_outside}, $\vec{r'}$ resolves the pair $(x,y)$. Ultimately, if $v=x$ or $v=y$, the pair $(x,y)$ is also resolved by $S'$ as $v \in S'$.

\noindent \textbf{(S2)} Let $\vec{r} \in \int {I}$. If $\vec{r}=(1,\ldots,1,0)$, as $\vec{Tr_{X_i}}(v)=\vec{r}$, there is a vertex in $S'$ with trace $\vec{r}$. Otherwise, as $I_1$ is compatible, there exists $\vec{r_1}\in \int{I_2}$ such that $\vec{r}=\vec{r_1|1}$. Let $s \in S$ such that $\vec{Tr_{X_j}}(s)=\vec{r_1}$, then $\vec{Tr_{X_i}}(s)=\vec{r}$. Indeed, the vertex $v$ is not the closest vertex of $X_i$ from $s$. If $s \in X_j$, that's clear. Otherwise $X_j$ is a separator between $s$ and $v$, so the shortest path between $s$ and $v$ crosses a vertex of $X_j$. Thus, $\vec{Tr_{X_i}}(s)=\vec{r_1|1}$.

\noindent \textbf{(S3)} Let $(\vec{r},\vec{t}) \in \paire I$. Let $x \in V(T(X_i))$ such that $\distvec{X_i}{x}=\vec{r}$ and $y \notin T(X_i)$ such that $\distvec{X_i}{y}=\vec{t}$. Then  $\distvec{X_j}{x}=\vec{r^-}$ and $\distvec{X_j}{y}=\vec{t^-}$ so the pair $(x,y)$ is resolved by $S$ because $(\vec{r^-},\vec{t^-})$ belongs to $\paire{I_1}$.

\noindent \textbf{(S4)} As $S_I=S_{I_1} \cup \{v\}$ we have $S \cap X_i=S_{I_1} \cup \{v\}=S_I$.\qed
\end{proof}

\begin{lemme}\label{join_node3}
Let $I$ be an instance for an introduce node $i$. 
Let $\mathcal{F}_{1}(I)$ be the set of instances $I_1$ for $i_1$ compatible with $I$ of type 1 and $\mathcal{F}_{2}(I)$ be the set of instances $I_2$ for $i_1$ compatible with $I$ of type~2. Then,
\[\dim(I) \leq \min\;\{ \min_{I_1 \in \mathcal{F}_{1}(I)}\;\{\dim(I_1)\},\min_{I_2 \in \mathcal{F}_{2}(I)}\;\{\dim(I_2)+1\}\}.\]
\end{lemme}
\begin{proof}
The proof directly follows from the fact that, for any instance $I_1$ for $j$ compatible with $I$, we can get a solution of $I$ of size $\dim(I_1)$ if $I_1 \in \mathcal{F}_1$ by Lemma~\ref{join_node1} and of size $\dim(I_1)+1$ if $I_1 \in \mathcal{F}_2$ by Lemma~\ref{join_node2}. \qed
\end{proof}
\begin{lemme}\label{join_node4}
Let $S$ be a solution for $I$  with $v \notin S$. Then there exists $I_1 \in \mathcal{F}_1$ such that $S$ is a solution of $I_1$.
\end{lemme}
\begin{proof}

Let $I_1$ be the instance defined by $S_{I_1}=S_I$, 
$\int{I_1}= (\cup_{\vec{r} \in \int{I}}\vec{r^-})$,
$\ext{I_1}=(\cup_{\vec{r} \in \ext{I}}\vec{r^-})$ and 
$\paire{I_1}= \cup_{(\vec{r},\vec{t}) \in \paire{I}} (\vec{r^-},\vec{t^-}))$.
One can easily remark that $I_1$ is compatible with $I$ of type $1$.

We prove that $S$ is a solution of $I_1$.

\noindent \textbf{(S1)} Let $(x,y)$ be a pair of vertices of $T(X_j)$. As $S$ is a solution for $I$, either there exists $s \in S$ that resolves the pair $(x,y)$, or there is a vector $\vec{r} \in \ext I$ that resolves $(x,y)$. In the second case, by construction of $I_1$, the vector $\vec{r^-}$ belongs to $\ext {I_1}$ and resolves $(x,y)$ by Lemma~\ref{resolved_outside}. So the pair $(x,y)$ is resolved in both cases. 

\noindent \textbf{(S2)} Let $\vec{r} \in \int {I_1}$. By construction, there is $\vec{r'} \in \int I$ such that $\vec{r'^-}=\vec{r}$. Let $s \in S$ such that $\vec{Tr_{X_i}}(s)=\vec{r'}$, then $\vec{Tr_{X_j}}(s)=\vec{r}$.

\noindent \textbf{(S3)} Let $(\vec{r},\vec{t}) \in \paire {I_1}$, $x \in \T {X_i}$ with $\distvec{X_j}{x}= \vec{r}$ and  $y \notin \T {X_j}$ with $\distvec{X_j}{y}= \vec{t}$ with $d(x,X_j) \leq 2$ and $d(y,X_j) \leq 2$. 
Let $(\vec{r'},\vec{t'}) \in \paire{I}$ such that $(\vec{r},\vec{t})=(\vec{r'^-},\vec{t'^-})$.
First $d(x,v)=d(x,X_j)+1$ because $X_j$ separates $x$ and $v$. This is true for any vertex of $X_j$ so the last component of $\vec{r'}$ is $d(x,X_j)+1$. So $\distvec{X_i}{x}= \vec{r'}$. If $\distvec{X_i}{y}= \vec{t'}$, then $(x,y)$ is resolved by a vertex of $s$. Otherwise, as $(\vec{r'},\vec{t'}) \in \paire{I}$, there exist a vertex $z \notin T(X_{j})$ such that $\distvec{X_i}{z}= \vec{t'}$ and $s$ is $S$ that resolves the pair $(x,z)$. Then $s$ resolves the pair $(x,y)$ because $d(s,y)=d(s,z)$.

\noindent \textbf{(S4)} We have $S_I=S_{I_1}$ and $v \notin S_I $ so $S \cap X_{j} =S_{I_1}$.

Finally, $S$ is a solution of $I_1$ so $\dim(I_1) \leq |S| \leq \dim(I)$. In particular $\dim(I) \geq  \min_{I_1 \in \mathcal{F}_1}\;\{\dim(I_1)\}$. \qed
\end{proof}
\begin{lemme}\label{join_node5}
Let $S$ be a solution for $I$ of minimal size with $v \in S$. Then there exists $I_2 \in \mathcal{F}_2$ such that $S \setminus \{v\}$ is a solution of $I_2$.
\end{lemme}
\begin{proof}
    
Let $I_2$ be the instance where $S_{I_1}=S_I \setminus \{v\}$, 
$\int{I_2}= (\cup_{\vec{r} \in \int{I}}\vec{r^-}) $,
$\ext{I_2}=\{\cup_{\vec{r} \in \ext{I}}\vec{r^-}\} \cup \{(0,\ldots0)\}$ and 
$\paire{I_2}= \cup_{(\vec{r},\vec{t}) \in \paire{I}} (\vec{r^-},\vec{t^-})$.
One can easily remark that $I_2$ is compatible with $I$ of type $2$.

We prove that $S$ is a solution of $I_2$.
\noindent \textbf{(S1)} Let $(x,y)$ be a pair of vertices of $T(X_j)$. As $S$ is a solution for $I$, either there exists $s \in S$ that resolves the pair $(x,y)$ or there is a vector $\vec{r} \in \ext I$ that resolves $(x,y)$. If $(x,y)$ is resolved by a vertex of $S \setminus \{v\}$, then the pair $(x,y)$ is resolved in $I_2$. If $(x,y)$ is resolved by $v$, then the vector $(0,\ldots0)$ of $\ext{I_2}$ resolves the pair. If $(x,y)$ is resolved by a vector $\vec{r}$ of $\ext I$, then by Lemma~\ref{resolved_outside}. So $\vec{r^-}$ resolves the pair $(x,y)$ and $\vec{r^-} \in \ext {I_2}$ by construction.

\noindent \textbf{(S2)} Let $\vec{r} \in \int {I_1}$. By construction, there is $\vec{r'} \in \int I$ such that $\vec{r'^-}=\vec{r}$. Let $s \in S$ such that $\vec{Tr_{X_i}}(s)=\vec{r'}$, then $\vec{Tr_{X_j}}(s)=\vec{r}$.

\noindent \textbf{(S3)} Let $(\vec{r},\vec{t}) \in \paire {I_1}$, $x \in \T {X_i}$ with $\distvec{X_j}{x}= \vec{r}$ and  $y \notin \T {X_j}$ with $\distvec{X_j}{y}= \vec{t}$ with $d(x,X_j) \leq 2$ and $d(y,X_j) \leq 2$. 
Let $(\vec{r'},\vec{t'}) \in \paire{I}$ such that $(\vec{r},\vec{t})=(\vec{r'^-},\vec{t'^-})$.
First $d(x,v)=d(x,X_j)+1$ because $X_j$ separates $x$ and $v$. This is true for any vertex of $X_j$ so the last component of $\vec{r'}$ is $d(x,X_j)+1$. So $\distvec{X_i}{x}= \vec{r'}$. If $\distvec{X_i}{y}= \vec{t'}$, then $(x,y)$ is resolved by a vertex of $s$. Otherwise, as $(\vec{r'},\vec{t}) \in \paire{I}$, there exist a vertex $z \notin T(X_j)$ such that $\distvec{X_i}{y}= \vec{t'}$ and $s$ in $S$ that resolves the pair $(x,z)$. Then $s$ resolves the pair $(x,y)$ because $d(s,y)=d(s,z)$.

\noindent \textbf{(S4)} We have $S_I=S_{I_1}$ and $v \notin S_I $ so $S \cap X_{j} =S_{I_1}$.

Finally, $S \setminus \{v\}$ is a solution of $I_2$, thus $\dim(I_2) \leq |S-1| \leq \dim(I)$. In particular $\dim(I) \geq  \min_{I_2 \in \mathcal{F}_2}\;\{\dim(I_2)+1\}$.\qed
\end{proof}
Lemma~\ref{node_introduce_main} is a consequence of Lemmas~\ref{join_node3},~\ref{join_node4} and~\ref{join_node5}.

\subsubsection{Forget node}

We now consider an instance $I$ for an forget node $i$. Let $j$ be the child of $i$ and $v \in V$ be such that $X_j=X_i \cup \{v\}$. Let $X_j=\{v_1,\ldots,v_k\}$ with $v=v_k$.
The trees $T(X_i)$ and $T(X_j)$ contain the same vertices, the definition of compatibility gives conditions to have the same solution for $I$ and for an instance on the child node.

We introduce three functions on vectors representing how the trace of a vertex can be modified when one considers two separators that differ by one vertex.

\begin{defn}
Let $\vec{r}$ be any binary vector. We define the functions $f$, $f^-$ and $f^+$ which return a vector with one more component. The function $f^-$ is defined as $\vec{f^-(r)}= \vec{r | \min(r_i)}$ if $\vec{r}$ is not constant and $\vec{f^-(r)}= \vec{ r| (r_1-1)}$ if $\vec{r}$ is constant. We define $f^+$ as $\vec{f^+(r)}=  \vec{r| \max(r_i)}$ is $\vec{r}$ is not constant and $\vec{f^+(r)}= \vec{ r| (r_1+1)}$ if $\vec{r}$ is constant.  We define $f$ as $\vec{f(r)}=  \vec{r| \min(r_i)}$ is $\vec{r}$ is not constant and $\vec{f(r)}= \vec{ r| r_1}$ if $\vec{r}$ is constant.
\end{defn}

The function $f$ is introduced only to deal with the case of constant vector.
These functions are defined to deal with the following case. Let $X_i$ and $X_j$ be two bags such that $X_i=X_j \setminus \{v\}$ for some vertex $v$.  Let $x$ be any vertex, then $\distvec{X_j}{x}$ is equal to $\vec{f(\distvec{X_i}{x})}$, $\vec{f^+(\distvec{X_i}{x})}$ or $\vec{f^-(\distvec{X_i}{x})}$. Moreover, if $X_i$ separates $x$ and $v$, $\distvec{X_j}{x}=\vec{f^+(\distvec{X_i}{x})}$.

\begin{defn}\label{compatible_forget}
Let $I$ be an instance for a forget node $i$ and let $j$ be the child of $i$ and $v \in V$ such that $X_i=X_j \setminus \{v\}$. Let $X_j=\{v_1,\ldots,v_k\}$ with $v=v_k$. An instance $I_1$ for $j$ is compatible with $I$ if 
\begin{itemize}
    \item \textbf{(F1)} $S_I=S_{I_1} \setminus \{v\}$.
    \item \textbf{(F2)} For all $ \vec{r} \in \ext{I}$,  $\vec{r|1} \in \ext{I_1}$.
    \item \textbf{(F3)} For all $ \vec{r} \in \int{I}$, $\vec{r|0}\in \int{I_1}$ or $\vec{r|1}\in \int{I_1}$. 
    \item \textbf{(F4)} $\forall (\vec{r},\vec{t}) \in \paire {I}$, if there exist two vertices $x \in \T {X_i}$ with $\distvec{X_j}{x}= \vec{f^-(r)}$ (resp. $\vec{f(r})$, $(\vec{f^+(r)}$)  and $y \notin \T {X_i}$ with $\distvec{X_j}{y}= \vec{f^+(t)}$ with $d(x,X_j) \leq 2$ and $d(y,X_j) \leq 2$ then $(\vec{f^-(r)},\vec{f^+(t)})$ (resp.  
    $(\vec{f(r}),\vec{f^+(t)})$, $(\vec{f^+(r)},\vec{f^+(t)})$) belongs to $ \paire {I_1}$.

\end{itemize}
\end{defn}

\begin{lemme}\label{node_forget_main}
    
Let $I$ be an instance for a forget node $i$. 
Let $\mathcal{F}_{F}(I)$ be the set of instances $I_1$ for $j$ compatible with $I$. Then, \[\dim(I)= \min_{I_1 \in \mathcal{F}_{F}(I)}\;\{\dim(I_1)\}.\]
\end{lemme}

The end of this section is devoted to prove Lemma~\ref{node_forget_main} by proving both inequalities in a similar way than for join and introduce nodes.

We prove a technical lemma similar to Lemma~\ref{resolved_outside} with a similar proof.

\begin{lemme}\label{resolved_outside2}
Let $i$ be a forget node, $j$ be the child of $i$ and $v \in V$ such that $X_i=X_j \setminus \{v\}$.
Let $(x,y)$ be a pair of vertices of $ \T{X_j}$. Let $\vec{r}$ be a binary vector of size $k$ with $\vec{r}_k=1$. Then $\vec{r}$ resolves $(x,y)$ if and only if $\vec{r^-}$ resolves $(x,y)$.
\end{lemme}
\begin{proof}
Let $\vec{t_1}=\vec{Tr_{X_i}}(x)$ and $\vec{t_2}=\vec{Tr_{X_i}}(y)$.
Assume $\vec{r}$ resolves $(x,y)$ and by contradiction that $\vec{r^-}$ does not resolve $(x,y)$. As $\vec{r}$ resolves $(x,y)$ we have  by Definition~\ref{def_resolve_vec}, $\min_{1 \leq l \leq k} \vec{(t_1+r)}_l \neq \min_{1 \leq l \leq k} \vec{(t_2+r)}_l$ and $\min_{1 \leq l \leq k-1} \vec{(t_1+r)}_l \neq \min_{1 \leq l \leq k-1} \vec{(t_2+r)}_l$.
So the minimum change in at least one case, assume by symmetry that $\min_{1 \leq l \leq k} \vec{(t_1+r)}_l \neq \min_{1 \leq i \leq k-1} \vec{(t_1+r)}_l$. Thus, for $l<k$, we have $\vec{(t_1+r)}_l > \vec{(t_1+r)}_k$. We know $\vec{r}_k=1$, so, for $l<k$, we have $\vec{(t_1+r)}_l > 1$. That gives $\vec{r}=(1,\ldots,1)$ which contradicts the fact that ${X_j}$ separates $v$ from $x_1$, one vertex of $X_j$ is strictly closer to $x$ than $v$.
Assume $\vec{r^-}$ resolves $(x,y)$ and by contradiction that $\vec{r}$ does not resolve $(x,y)$. Then, by symmetry we can assume that $\min_{1 \leq l \leq k} \vec{(t_1+r)}_l \neq \min_{1 \leq l \leq k-1} \vec{(t_1+r)}_l$, meaning $ \vec{(t_1+r)}_k < \min_{1 \leq l \leq k-1} \vec{(t_1+r)}_l$. We know $\vec{r}_k=1$ so $\vec{(t_1+r)}_k \geq 1$ and $\vec{(t_1+r)}_l =2$ for $1 \leq l \leq k-1$. So $\vec{r}=(1,\ldots1)$ which contradicts the fact that ${X_j}$ separates $v$ from $s$.\qed
\end{proof}

\begin{lemme}\label{calcul_forget1}
Let $I_1 \in \mathcal{F}_{F}(I)$ and $S$ be a solution for $I_1$, then $S$ is a solution for~$I$.
\end{lemme}

\begin{proof}
Let us prove that the conditions of Definition~\ref{def_instance} are satisfied.

\noindent \textbf{(S1)} Let $(x,y)$ be a pair of vertices of $T(X_i)$. As $T(X_i)=T(X_j)$, the pair is resolved by a vertex of $S$ or by a vector $\vec{r}$ of $\ext{I_1}$. If $(x,y)$ is resolved by $\vec{r} \in \ext{I_1}$ then by Lemma~\ref{resolved_outside2}, $\vec{r^-}$ resolves the pair $(x,y)$ and $\vec{r^-} \in \ext {I}$ by compatibility.

\noindent \textbf{(S2)} Let $\vec{r} \in \int{I}$. By compatibility, $\vec{r|0} \in \int{I_1}$ or $\vec{r|1} \in \int{I_1}$. Let $s \in S$ such that $\vec{Tr_{X_i}}(s) \in \{ \vec{r|0},\vec{r|1} \}$, then $\vec{Tr_{X_j}}(s) =\vec{r}$.

\noindent \textbf{(S3)} Let $(\vec{r},\vec{t}) \in \paire{I}$, $x \in V(T(X_{i}))$ and $y\notin T(X_{i})$ such that $\distvec{X_{i}}{x}=\vec{r}$ and $\distvec{X_{i}}{y}=\vec{t}$. Assume also $d(x,X_i) \leq 2$ and $d(y,X_i) \leq 2$. The set $X_j$ separates $v$ and $y$ so $\distvec{X_j}{y}= f^+(t)$. As  $\distvec{X_j}{x}$ is equal either to $\vec{f^-(r)}$, $\vec{f(r)}$ or to $\vec{f^+(r)}$, the pair $(x,y)$ is resolved by a vertex of $S$.

\noindent \textbf{(S4)} is clear.\qed

\end{proof}
\begin{lemme}\label{calcul_forget2}
Let $S$ be a solution for $I$ of minimal size. Then there exists $I_1$ compatible with $I$ such that $S$ is a solution of $I_1$.
\end{lemme}

\begin{proof}
Let $S$ be an solution for $I$ of minimal size. Let $I_1$ be the following instance: $S_{I_1} = S \cap X_j$, $\int {I_1}= \{ \vec{Tr_{X_j}}(s),s \in S\}, \ext{I_1} = \{\vec{r}|1, \vec{r} \in \ext{I} \}, \paire {I_1} =\{ (\vec{f^-(r)},\vec{f^+(t)}), (\vec{r},\vec{t}) \in \paire{I} \} \cup \{ (\vec{f(r)},\vec{f^+(t)}), (\vec{r},\vec{t}) \in \paire{I} \} \cup \{ (\vec{f^+(r)},\vec{f^+(t)}) (\vec{r},\vec{t}) \in \paire{I} \}$.
We first check the compatibility.

\noindent \textbf{(F1)},  \textbf{(F2)} and  \textbf{(F4)} are straightforward.

\noindent \textbf{(F3)} Let $\vec{r} \in \int{I}$ and $s\in S$ such that $\vec{Tr_{X_i}}(s)=\vec{r}$. By construction,  $\vec{r'}=\vec{Tr_{X_j}}(s)$ belongs to $\int {I_1}$ and $\vec{r'}=\vec{r^-}$ so $\vec{r|0} \in \int{I_1}$ or $\vec{r|1} \in \int{I_1}$. 
\vspace{1mm}

We prove now $S$ is a solution for $I_1$.

\noindent \textbf{(S1)} Let $(x,y)$ be a pair of vertices of $T(X_j)$. As $T(X_i)=T(X_j)$, the pair is resolved by a vertex of $S$ or by a vector $\vec{r}$ of $\ext{I}$. If $(x,y)$ is resolved by $\vec{r} \in \ext{I}$ then by Lemma~\ref{resolved_outside2}, $\vec{r|1}$ resolves the pair $(x,y)$ and $\vec{r|1} \in \ext {I_1}$ by construction.

\noindent \textbf{(S2)} Let $\vec{r} \in \int{I_1}$. By construction there is $s \in S$ such that $\vec{Tr_{X_j}}(s)=\vec{r}$.

\noindent \textbf{(S3)} Let $(\vec{r},\vec{t}) \in \paire{I_1}$, $x \in V(T(X_j))$ and $y\notin T(X_j)$ such that $\distvec{X_j}{x}=\vec{r}$ and $\distvec{X_j}{y}=\vec{t}$. Assume also $d(x,X_i) \leq 2$ and $d(y,X_i) \leq 2$. Then $\vec{Tr_{X_j}}(y)=f^+(\vec{t})$ and  $\vec{Tr_{X_j}}(x) \in  \{ \vec{f^-(r)}, \vec{f(r)}, \vec{f^+(r)} \}$. Since $S$ is a solution of $I$, the pair $(x,y)$ is resolved by a vertex of $S$.

\noindent \textbf{(S4)} is clear.\qed
\end{proof}

Lemma~\ref{node_forget_main} is a consequence of Lemmas~\ref{calcul_forget1} and~\ref{calcul_forget2}.

\subsection{Algorithm}

Given as input a nice clique tree, the algorithm computes the extended metric dimension of all the possible instances bottom up from the leaves. The algorithm computes the values for leaves using Lemma~\ref{calcul_leaf}, for join nodes using Lemma~\ref{node_pair_main}, for introduce nodes using Lemma~\ref{node_introduce_main} and forget nodes using Lemma~\ref{node_forget_main}.
The correction of the algorithm is straightforward by these lemmas. 

We denote this algorithm by $IMD$ in the following which takes as input a nice clique tree $T$ and outputs the minimal size of a resolving set of $G$ containing the root of $T$.

\section{Proof of Theorem~\ref{thm:main}}\label{sec:complexity}

Let us finally explain how we can compute the metric dimension of $G$.

\begin{lemme} \label{multi}
The metric dimension of $G$ is $ \min_{v \in G} \{IMD(T(v))\}$ where $T(v)$ is a nice clique tree of $G$ rooted in $v$.
\end{lemme}

\begin{proof}
For any input, $IMD(T(v))$ outputs the size of a resolving set of $G$. So, $ \min_{v \in G} \{IMD(T(v))\} \geq \dim(G)$. Let $S$ be a minimum resolving set of $G$ and let $v \in S$. By Lemma~\ref{correct}, $IMD(T(v))$ outputs the minimum size of a resolving set containing $v$ so $ \min_{v \in G} \{IMD(T(v))\} \leq \dim(G)$ which complete the proof. \qed
\end{proof}

In particular, $n$ executions of the $IMD$ algorithm with different inputs are enough to compute the metric dimension. Lemma~\ref{tree_dec} ensures that we can find for any vertex $v$ of $G$ a nice clique tree in linear time, the last part is to compute the complexity of the $IMD$ algorithm.

To get the announced complexity, we add a first step to the $IMD$ algorithm: for each bag $X$, we compute $\ddeux{X} \cap T(X)$ and $\ddeux{X} \cap (G \setminus T(X))$. This computation can be done in $O(n^2)$ times (recall that $T$ has a linear number of bag by Lemma~\ref{tree_dec}). Note also that the size of $\ddeux{X} \cap T(X)$ and $\ddeux{X} \cap (G \setminus T(X))$ depends only of $|X_i|$.

To compute the complexity, we need to compute the number of instances and the time to solve an instance. For simplicity we let $\alpha(k):= 2^k\cdot 2^{2^k}\cdot 2^{2^k}\cdot 2^{4^{2k}}$.

\begin{lemme}\label{complexity_instance}
Let $I $ be any instance for a node $i$ and assume $\dim(I')$ is known for every instance $I'$ compatible with $I$ for every child of $i$ . Then $\dim (I)$ can be computed in time $O(f(|X_i|))$ for a computable function $f$.
\end{lemme}

\begin{proof}
If $i$ is a leaf node then $\dim (I)$ can be computed in constant time by Lemma~\ref{calcul_leaf}. Otherwise, let us prove that one can compute for all the instances on the child nodes (at most two child nodes) all the compatible instances.
Given a $5$-uplet  $(X_i,S_I, \int I,\ext I,\paire{I})$, checked if it is an instance according to Definition~\ref{def_instance} can be done in time $O(|I|\cdot g(|X_i|))$ where $g$ gives the size of $\ddeux{X_i} \cap T(X_i)$ plus the size $\ddeux{X_i} \cap (G \setminus T(X_i))$. The number of such $5$-uplet $(X_i,S_I, \int I,\ext I,\paire{I})$ is bounded by $\alpha(|X_i|)$. Thus, identifying the instances among all the $5$-uplet ca be done in time depending only of $|X_i|$.

Checking the compatibility can be done in a time that only depends on $|X_i|$. Condition \textbf{(J5)} can be checked in time $O(|X_i|^2 \cdot |I|)$ to check for each pair of vectors if a vector of $\ext{I}$ resolves it. 
Condition \textbf{(I5)} can be checked in time $O(|X_i|)$ and condition \textbf{(F4)} in time $O(|X_i|^2)$. The other compatibility conditions can be checked in time $O(|I|)$ and by Definition~\ref{def_instance}, $|I|$ is bounded by a function of $|X_i|$.
Then, computing the minimum using the formulas of Lemmas~\ref{node_pair_main},~\ref{node_introduce_main} and~\ref{node_forget_main} can be done in time $O(\alpha(X_i))$. Ultimately the computation of $\dim(I)$ is done in time bounded by a function of $X_i$. \qed
\end{proof}

\begin{cor}\label{complexiteIMD}
The algorithm for $IMD$ runs in time $O(n(T)^2 + n(T) \cdot f(\omega))$ where $n(T)$ is the number of vertices of the input tree $T$ and $f=O(k^2 \cdot 2^{O(4^{2^k})})$ is a function that only depends on the size of a maximum clique $\omega$.
\end{cor}

\begin{proof}
By definition of the treewidth, for any bag $X$ of $T$, $|X| \leq \omega$. The first step of computation to get $\ddeux{X_i} \cap T(X_i)$ and $\ddeux{X_i} \cap (G \setminus T(X_i))$ takes time $O(n^2)$. Then, the number of instances to compute for each vertex of $T$ is bounded by $\alpha(\omega)$ and each instance $I$ can be computed in time bounded $O(\omega^2 \cdot |I|)$ by Lemma~\ref{complexity_instance}. \qed
\end{proof}

We now have all the ingredients to prove Theorem~\ref{thm:main}:

\begin{proof}
For each vertex $v$ of $G$, one can compute a nice clique tree of size at most $7n$ according to Lemma~\ref{tree_dec}. Given this clique tree, the $IMD$ algorithm outputs the size of a smallest resolving set containing $v$ by Lemma~\ref{correct} in time $O(n(T)^2 + n(T) \cdot f(\omega))$ for a computable function $f$ according to Corollary~\ref{complexiteIMD}. Repeat this for all vertices of $G$ permits to compute the metric dimension of $G$ by Lemma~\ref{multi} in time $O(n^3+n^2 \cdot f(\omega))$.\qed
\end{proof}

\bibliography{biblio} 
\bibliographystyle{alpha} 

\end{document}